\documentclass[article,notitlepage,amsthm,amsmath,amssymb]{revtex4-1}

\usepackage{graphicx}
\usepackage{amsmath}
\usepackage{amssymb}
\usepackage{amsthm}
\usepackage{algpseudocode}
\usepackage{algorithm}
\usepackage{enumitem}
\usepackage{float}
\usepackage{hyperref}
\usepackage{tikz}
\usepackage{thm-restate}

\newcommand{\be}{\begin{equation}}
\newcommand{\ee}{\end{equation}}
\newcommand{\ba}{\begin{array}}
\newcommand{\ea}{\end{array}}
\newcommand{\bea}{\begin{eqnarray}}
\newcommand{\eea}{\end{eqnarray}}

\newcommand{\ra}{\rangle}
\newcommand{\la}{\langle}

\newcommand{\FF}{\mathbb{F}}

\newcommand{\Gate}[1]{\textsc{#1}}

\newcommand{\cnotgate}{\Gate{CNOT}} 
\newcommand{\czgate}{\Gate{CZ}}

\newcommand{\xgate}{\Gate{X}}

\newcommand{\zgate}{\Gate{Z}}

\newcommand{\eq}[1]{Eq.~(\ref{eq:#1})}
\renewcommand{\sec}[1]{\hyperref[sec:#1]{Section~\ref*{sec:#1}}}
\newcommand{\ssec}[1]{\hyperref[ssec:#1]{Subsection~\ref*{ssec:#1}}}
\newcommand{\appe}[1]{\hyperref[appe:#1]{Appendix~\ref*{appe:#1}}}
\newcommand{\fig}[1]{\hyperref[fig:#1]{Figure~\ref*{fig:#1}}}
\newcommand{\tab}[1]{\hyperref[tab:#1]{Table~\ref*{tab:#1}}}
\newcommand{\lem}[1]{\hyperref[lem:#1]{Lemma~\ref*{lem:#1}}}
\newcommand{\thm}[1]{\hyperref[thm:#1]{Theorem~\ref*{thm:#1}}}

\newcommand{\ket}[1]{{\left\vert{#1}\right\rangle}}

\newtheorem{lemma}{Lemma}
\newtheorem{fact}{Fact}

\newtheorem{theorem}{Theorem}

\begin{document}

\title{Constant-cost implementations of Clifford operations and multiply controlled gates using global interactions}

\author{Sergey Bravyi}
\affiliation{IBM Quantum, IBM T. J. Watson Research Center, Yorktown Heights, NY 10598, USA}
\author{Dmitri Maslov}
\affiliation{IBM Quantum, IBM T. J. Watson Research Center, Yorktown Heights, NY 10598, USA}
\author{Yunseong Nam}
\affiliation{Department of Physics, University of Maryland, College Park, MD 20742, USA}

\begin{abstract}
We consider quantum circuits composed of single-qubit operations
and global entangling gates generated by  Ising-type Hamiltonians.
It is shown that such circuits can implement 
 a large class of unitary operators commonly used in quantum algorithms
at a very low cost—using a constant or effectively constant number of global entangling gates.
Specifically,  we report 
constant-cost implementations of Clifford operations with and without ancillae, constant-cost implementation of the multiply controlled gates with linearly many ancillae, and an $O(\log^*(n))$ cost implementation of the $n$-controlled single-target gates using logarithmically many ancillae.  This shows a significant asymptotic advantage of circuits enabled by the global entangling gates.
\end{abstract}

\maketitle

\section{Introduction}\label{sec:Intro}

Executing a quantum computer program requires its compilation to native gates directly supported by the controlling apparatus. While most of the existing quantum computing hardware supports arbitrary single-qubit gates as their native physical-level instruction, the choice of entangling multi-qubit operations may vary strongly across different platforms. Here we focus on leveraging parallel global instructions available in some quantum computer architectures~\cite{nigg2014quantum, grzesiak2020efficient, bassler2022synthesis}.
Such global operations draw their power from the ability to apply a layer of pairwise commuting two-qubit gates simultaneously. A notable example is Molmer–Sorensen gate available in the trapped-ion architecture~\cite{sorensen1999quantum,grzesiak2020efficient} that evolves a system of qubits under the Ising-type Hamiltonian. 
In this work, we report a streamlined implementation of two common quantum computing subroutines enabled by global entangling gates: elements of the Clifford group and multiply-controlled gates.  Clifford operations play a central role in quantum error correction, certain simulation 
algorithms~\cite{van2020circuit}
and a variety of applications based on 
pseudo-random unitary operators~\cite{huang2020predicting,haferkamp2020quantum,elben2022randomized}.
Multiply-controlled gates are
ubiquitous in applications that rely on
arithmetic circuits or the quantum singular value
transformation~\cite{gilyen2019quantum,martyn2021grand}.

We consider computational primitives that implement a selectable set of commuting two-body Ising interactions in one go, which we call global tunable gates.  We furthermore treat single-qubit gates as a free resource, which is justified since the implementation of entangling operators on the physical level tends to be more difficult and error-prone compared to the single-qubit gates.  Such a selection of the costing metric implies that the particular nature of the Ising interaction (e.g., XX, YY, ZZ) used is irrelevant---indeed, the results hold uniformly for all possible choices of the interaction since the changes between them are accomplished by the single-qubit gates (basis change) and can be merged into the single-qubit gate layers.  For the sake of simplicity of mathematical expressions, we chose to work with $\czgate^a$ interaction defined as
$\czgate^a\,{=} e^{i \pi a|11\ra\la 11|}$. Here $a{\in} [0,1]$ is a tunable parameter.  In other words, $\czgate^0$ is equivalent to the identity, whereas $\czgate^1\,{=}\,\czgate$ is single-qubit equivalent to the $\cnotgate$ gate.
We define an $n$-qubit global tunable (GT) gate as a 
diagonal unitary operator
\[
\prod_{1\le i<j\le n} \czgate^{a(i,j)}_{i,j},
\]
where $i$ and $j$ indicate pairs of qubits
acted upon by each $\czgate^a$ gate
and we allow the parameter $a{=}a(i,j)$
to be tuned individually for each pair of qubits $i$ and $j$. We assume all-to-all qubit connectivity.

Multiple versions of GT gates have been experimentally demonstrated~\cite{nigg2014quantum, figgatt2019parallel, landsman2019two, grzesiak2020efficient}. 
Shown in particular in \cite{grzesiak2020efficient} was the pulse complexity required for a GT gate is not much more than that required for a single two-qubit gate, i.e., $3n{-}1$ degrees of freedom used for a GT instruction and $2n{+}1$ for a two-qubit gate, for an $n$-qubit quantum computer. This is understood by considering the pulse design space, where some number of degrees of freedom must first be spent to decouple the computational states from the mutually shared information bus at the end of the entangling operation ($2n$) and then using additional degrees of freedom to induce the desired degree of entanglement between qubit pairs of choice. Since there are at most $O(n^2)$ pairs to entangle, $O(n)$ additional degrees of freedom introduced to each of $n$ qubits suffice, making the hardware-level requirement for a GT gate and that for a two-qubit gate not too different.

Our main contributions are as follows. First, we show that any $n$-qubit Clifford operation can be implemented using a constant number of GT gates (at most $21$) using no ancilla. The number of GT gates is at most four if $n$ ancillae are available. The latter construction meets the information-theoretic lower bound that applies to the no-ancilla case.  
Second, we show how to construct a multiply-controlled Toffoli (Toffoli-$n$) gate using $O(\log^*(n))$ \footnote{Here
$\log^*(n)$ denotes the iterated logarithm of $n$  defined as the smallest integer $p$
such that $\log{ \log{\cdots \log{(n)}}}\le 1$, where the logarithm is taken $p$ times.} GT gates with $O(\log(n))$ ancillae or using $4$ GT gates with $O(n)$ ancillae.
We also consider
{\em adaptive} quantum circuits that include
intermediate measurements and classical feedforward.
We show that Toffoli-$n$ can be implemented 
by an adaptive circuit with only two GT gates
and $O(n)$ ancillae. Our results achieve a substantial improvement over the state-of-the-art. 

The rest of the paper is organized as follows. \sec{prevwork} summarizes known relevant compilation results.  \sec{cliff} and \sec{mct} report our main results---explicit circuit constructions for $n$-qubit Clifford operators and the multiply-controlled Toffoli.  We finish with a discussion of our findings and further outlook in \sec{disc}.

\section{Previous work}\label{sec:prevwork}

The use of the GT gates to efficiently implement an $n$-qubit Clifford operation has been studied in the literature.  An implementation with $12n\,{+}\,O(1)$ GT gates requiring no ancilla was developed in \cite{maslov2018use}, which was superseded by \cite{van2021constructing} that improved it to $6n\,{+}\,O(1)$ GT gates, still using no ancilla.  Subsequently in \cite{grzesiak2022efficient} the complexity was improved to $6\log(n)\,{+}\,O(1)$ GT gates, however, with $n/2$ ancillae.  Recently in \cite{bassler2022synthesis}, an ancilla-free, $(n,2n]$-GT method was reported.  Finally, an ancilla-free implementation with $2\log(n)\,{+}\,O(1)$ GT gates follows from \cite{maslov2022depth} by noticing that $R_A^{-1}W01$ in formula (4) therein
is a single GT gate and non-overlapping GT gates can be implemented in parallel.  In all these approaches, one employs a decomposition of an $n$-qubit Clifford into stages of single- and two-qubit gate layers, to then implement the two-qubit gate layers efficiently using GT gates.  Here, we too employ this 
decomposition, see e.g. Lemma~8 of \cite{bravyi2021hadamard}, where any $n$-qubit Clifford $U$ can be realized by a layered circuit
\be
\label{eq:Clifford}
\mbox{$U$ = -L-CX-CZ-L-CZ-L-}.
\ee
Here -L- stands for a layer of single-qubit gates, -CX- is a linear reversible circuit, and -CZ- is a layer of $\czgate$ gates.  In contrast to previous results, we show how to implement arbitrary Clifford operation using $O(1)$ GT gates either with or without ancilla (the constant is smaller with ancilla). We note that a single GT gate suffices if the input state of a Clifford circuit is a fixed basis vector, such as $|0^n\ra$. Indeed, in this case, the output state $U|0^n\ra$ can be prepared by a layered circuit -L-CZ-L-, as follows from the well-known local equivalence between stabilizer states and graph states~\cite{schlingemann2001stabilizer}. In contrast, our implementation with $O(1)$ GT gates works for an arbitrary input state.

GT gates were also used to efficiently implement a Toffoli-$n$, implying same asymptotic for arbitrary multi-controlled unitaries.  Briefly, a $3$-GT construction of Toffoli-3 was reported in \cite{martinez2016compiling} and a $7$-GT construction of Toffoli-4 was reported in \cite{ivanov2015efficient}.  Toffoli-4 was improved to rely on $3$ GTs in \cite{maslov2018use}, with one clean ancilla, along with  $3n\,{+}\,O(1)$ GT implementation of a multiply-controlled Toffoli using $n/2\,{+}\,O(1)$ ancillae.  An ancilla-free implementation (of, technically speaking, multicontrolled $i\xgate$) that takes $2n$ GT gates was reported in \cite{groenland2020signal}.  In all of these references, the GT gates used were the single-angle, all-pair (drawn from a selectable subset of qubits) kind.  Indeed, when using a more general GT, the complexity can be improved to $O(\log(n))$ by nesting Toffoli-$m$, $m\leq n$, in a depth-optimal fashion with $O(n)$ ancillae, or $3n/2$ GTs with no more than seven ancillae~\cite{grzesiak2022efficient}.

\section{Clifford operations}\label{sec:cliff}

We focus on implementing $n$-qubit Clifford operations by quantum circuits composed of GT gates and ``free" single-qubit gates.  Here, we use only Clifford GT gates that apply $\czgate$ to a subset of qubits.  Following Ref.~\cite{van2021constructing} we will refer to such gates as Global $\czgate$ (GCZ).  Our main result is the following. 
\begin{theorem}
\label{thm:main}
Any $n$-qubit Clifford operator can be implemented using $4$ GCZs with $n$ ancillae
or using $26$ GCZs with no ancilla.
\end{theorem}
To prove the theorem we first show how to realize
a special class of Clifford operators using $3$ GCZs.
Let $\mathrm{GL}(n)$ be the group of binary $n\,{\times}\,n$ invertible matrices.  Here and below multiplication and inversion for binary matrices is defined over the binary field $\FF_2$.
\begin{lemma}
\label{lem:1}
For any matrix $A\,{\in}\,\mathrm{GL}(n)$ there exist $2n$-qubit Clifford circuits $C_3(A)$ and $C_3'(A)$, each with the GCZ cost of $3$, such that 
\be
\label{eq:C3}
C_3(A) |x,y\ra = |A^{-1} y,Ax\ra \quad \mbox{and}
\ee
\be
\label{eq:C3'}
C_3'(A) |x,y\ra = |y\oplus Ax, A^{-1} y\ra
\ee
for all $x,y\in \FF_2^n$. These circuits require no ancilla. 
\end{lemma}
\begin{proof}
Given a binary $n\,{\times}\,n$ matrix $M$ and a pair of $n$-qubit registers $R_1$ and $R_2$, let $\mathsf{CX}_{1,2}(M)$ be a $2n$-qubit unitary operator that EXORs the $i$-th qubit of $R_1$ to the $j$-th qubit of $R_2$ for each pair pair $i,j\in \{1,2,\ldots,n\}$ such that $M_{i,j}\,{=}\,1$.
Here and below EXOR stands for the Exclusive OR.
 By definition we have 
$\mathsf{CX}_{1,2} (M)|x,y\ra = |x,y\oplus Mx\ra$
and  $\mathsf{CX}_{2,1} (M)|x,y\ra = |x\oplus My,y\ra$
for all $x, y\in \FF_2^n$. 
These operators can be implemented with a single GCZ by conjugating
each qubit in the register $R_2$ and $R_1$ respectively by the Hadamard gate
(which converts all $\cnotgate$s to $\czgate$s) and noticing that any layer of $\czgate$s 
is a single GCZ gate.
Now the desired transformations in Eqs.~(\ref{eq:C3},\ref{eq:C3'}) can be implemented as
\[
C_3(A) =  \mathsf{CX}_{1,2}(A) \, \mathsf{CX}_{2,1}(A^{-1}) \,  \mathsf{CX}_{1,2}(A)
\]
and 
\[
C_3'(A) = \mathsf{CX}_{1,2}(I{\oplus}A^{-1}) \,  \mathsf{CX}_{2,1}(I) \, \mathsf{CX}_{1,2}(I {\oplus} A).
\]

\end{proof}
We are now ready to prove Theorem~\ref{thm:main}.
\begin{proof}
According to \eq{Clifford}, it takes two -CZ- layers and one -CX- layer to implement arbitrary $n$-qubit Clifford operation $U$.  This implementation uses no ancilla.  Each -CZ- layer is, trivially, a single GCZ gate.  Thus below we focus on the -CX- layer.
Let $V$ be the circuit implementation of the -CX- layer in \eq{Clifford}.  We have  
$V|x\ra = |Ax\ra$
for all $x\,{\in}\, \FF_2^n$  and some matrix $A\,{\in}\, \mathrm{GL}(n)$.  
Using \lem{1} with $y\,{=}\,0^n$ one concludes that the map $\ket{x,0^n} \mapsto \ket{Ax, 0^n}$ can be implemented with $3$ GCZs, see \eq{C3'}.  Equivalently, $V$ admits a  $3$-GCZ implementation using $n$ ancillae initialized to and returned in the state $|0\ra$.  This shows that $U$ admits a $5$-GCZ implementation using $n$ ancilla.  Moreover, the left -CZ- layer in \eq{Clifford} can be merged  with the rightmost GCZ gate in the implementation of $V$.  Indeed, this -CZ- layer acts non-trivially only on the data/top $n$ qubits and thus it commutes with the layer of Hadamards in the implementation of $V$ which acts non-trivially only on the ancillary/bottom $n$ qubits.  Thus $U$ admits a $4$-GCZ implementation using $n$ ancilla.  This proves the first part of the theorem.

To prove the second part, we develop an ancilla-free implementation with at most $25$ GCZs when $3|n$ ($3$ divides $n$) and $26$ GCZs otherwise.  Based on \eq{Clifford}, it suffices to show that any $n$-qubit linear reversible circuit $V$ can be implemented using at most $23$ GCZs with zero ancilla (when $3|n$).  

Let us first describe a $12$-GCZ implementation of a $3k$-qubit linear reversible circuit $W(A)$ that uses $2k$ {\em dirty} ancillae (that is, ancillae that store an unknown data and must be returned in the initial state) to implement a $k$-qubit linear reversible transformation $\ket{x} \,{\mapsto}\, \ket{Ax}$ up to a register swap.  Specifically, suppose the initial state of $k$ data qubits and $2k$ ancillae is $\ket{x,y,z}$.  We would like to implement the map
\be
\label{eq:VII}
W(A)|x,y,z\ra=|z,Ax,y\ra.
\ee

We will use the following fact proved in~\cite{thompson1960commutators}.
\begin{fact}
Suppose $n\,{\ge}\,3$. Then any matrix $A\,{\in}\, \mathrm{GL}(n)$ can be expressed as a group commutator $A=D^{-1}B^{-1}DB$ for some $B,D\in \mathrm{GL}(n)$.
\end{fact}

Apply the circuit $C_3(A)$ constructed in \lem{1} four times with the matrix $A$ replaced by $B$ (first application), $D$ (second application), $B^{-1}$ (third application), and $D^{-1}$ (fourth application).  This gives the sequence of transformations
\begin{align*}
|x,y,z\ra &\xrightarrow{C_3(B)}
|B^{-1}y,Bx,z\ra \xrightarrow{C_3(D)}
|B^{-1}y, D^{-1} z, DBx \ra \\ 
& \xrightarrow{C_3(B^{-1})} 
|B^{-1} DBx, D^{-1} z,y\ra\\
&\xrightarrow{C_3(D^{-1})} 
|z,D^{-1} B^{-1} DBx, y\ra
= |z,Ax,y\ra.
\end{align*}
In other words, we implemented the $k$-qubit transformation $\ket{x} \mapsto \ket{Ax}$ using $2k$ dirty ancilla up to qubit register permutation using four applications of $C_3(\cdot)$, which costs $12({=}\,4{\cdot}3)$ GCZs.  

Assume that $3|n$.  We want to implement the map $|v\ra \,{\mapsto}\, |Av\ra$ with a given $A{\in} \mathrm{GL}(n)$ using constantly many GCZ gates and no ancilla.  Specifically, write $A$ as
\begin{equation}\label{eq:3x3}
A = \left[ \ba{ccc}
A_{00} & A_{01} & A_{02} \\
A_{10} & A_{11} & A_{12} \\
A_{20} & A_{21} & A_{22} \\
\ea
\right].
\end{equation}
Here each block $A_{i,j}$ has size $k$, where $k{=}n/3$, which is well-defined given $3|n$. First, we use two GCZ gates to transform $A$ into a form where $A_{00}$ and $\left(\begin{smallmatrix}A_{00} & A_{01}\\A_{10} & A_{11} \end{smallmatrix}\right)$ are invertible.

To accomplish this, label the individual matrix rows as $x_1,x_2,...,x_k,y_1,y_2,...,y_k,z_1,z_2,...,z_k$ top to bottom. Consider the first $k$ elements of these rows. We need to make $x_1,x_2,...,x_k$ restricted to the first $k$ bits be linearly independent.  This can be done by adding a row from the $y_{*}$ or $z_{*}$ sets onto some of $x_{*}$ rows.  This is accomplished in parallel by a depth-1 $\cnotgate$ circuit made with $\cnotgate(y_i,x_j)$ or $\cnotgate(z_i,x_j)$ gates. The $x_{*}$ rows are now linearly independent as prefix rows of length $k$, and as a result also as prefix rows of length $2k$. Thus, to make the first $2k$ rows linearly independent as length-$2k$ prefix rows, we can add some $\cnotgate(z_h,y_i)$ in the $\cnotgate$ depth 1.  Composing both stages we obtain the transformation of $A$ into the desired form at the GCZ cost of two, since each $\cnotgate$ layer can be implemented as a GCZ cost-1 operation.  


Next, we use $3$ GCZ gates to transform $A$ into the block-diagonal form such that the only non-zero blocks are $A_{00}$, $A_{11}$, and $A_{22}$.
First, apply the transformation 
\begin{equation}
\label{block_diag2}
A \gets 
\left[ \ba{ccc}
I & 0 & 0 \\
A_{10}A_{00}^{-1} & I & 0 \\
A_{20}A_{00}^{-1} & 0 & I \\
\ea
\right]\cdot A.
\end{equation}
It sets to zero the blocks $A_{10}$ and $A_{20}$ (and, possibly, modifies the blocks $A_{11}$, $A_{12}$, $A_{21}$, and $A_{22}$). This transformation can be implemented by the $\cnotgate$ gates with controls on the top third of qubits and targets on the bottom two thirds and thus it costs one GCZ gate (and some extra Hadamard gates). Since we assumed that $\left(\begin{smallmatrix}A_{00} & A_{01}\\A_{10} & A_{11} \end{smallmatrix}\right)$ is invertible, after the transformation Eq.~(\ref{block_diag2}) the block $A_{11}$ is invertible.
Thus the same argument applies to the remaining off-diagonal blocks $A_{01}$ and $A_{21}$, and then again for the pair of blocks $A_{02}$ and $A_{12}$.  Transforming $A$ to the block-diagonal form 
\be
\label{eq:Adiag}
A \gets \left[ \ba{ccc}
A_{00} & 0 & 0 \\
0 & A_{11} & 0 \\
0 & 0 & A_{22} \\
\ea
\right]
\ee
thus costs $5$ GCZs.  Note that each diagonal block describes a linear reversible circuit acting non-trivially on exactly $n/3$ qubits.
To implement all three, we can combine $C_3$ from \lem{1} twice (first time applied to first and third multiregisters, second time applied to first and second multiregisters) with the 12-GCZ circuit implementing the transformation in \eq{VII} for $k\,{=}\,n/3$ as follows:
\begin{align*}
|x,y,z\ra & \xrightarrow{C_3(A_{00})} |A_{00}^{-1}z, y, A_{00}x \ra \\ & \xrightarrow{C_3(A_{22}A_{00})} 
|(A_{22}A_{00})^{-1}y, A_{22}z, A_{00}x\rangle \\
& \xrightarrow{W(A_{11}A_{22}A_{00})} 
|A_{00}x,A_{11}y, A_{22}z\ra.
\end{align*}
The overall GCZ cost of $A$ is $5\,{+}\,3{\cdot}2\,{+}\,12 = 23$, using zero ancilla.  Thus the overall cost of the Clifford $U$ when $3|n$ is $25$ GCZs. 

Finally, consider the case $3{\not|}n$.  Suppose $n\,{=}\,3k{+}2$ ($n\,{=}\,3k{+}1$ is handled similarly).  Then, in \eq{3x3} the qubits are broken into the sets of $k$, $k{+}1$, and $k{+}1$.  To employ the circuit $W(A)$ in \eq{VII}, we have one less qubit than we need. To compensate for this lack of a qubit, an intermediate step is used after the reduction to block diagonal form, \eq{Adiag}, where, at the cost of one GCZ gate, we disentangle/diagonalize one qubit each in the two qubit sets of the size $k{+}1$.  Then we have a block diagonal decomposition with five blocks, of which three have sizes $k\,{\times}\, k$, and the remaining two have sizes $1{\times} 1$. The circuit $W(A)$ can now be implemented since there is enough ancillary space to accommodate it.  The overall cost of the Clifford $U$ for arbitrary $n$ is thus at most $26$ GCZs.

In \appe{A} we describe 
local optimizations reducing the cost further to $20$ or $21$ GCZs, depending on whether $3|n$.
\end{proof}

\section{Multiply-Controlled Toffoli}\label{sec:mct}

We show two different methods to implement an $n$-qubit Toffoli gate. The first uses $O(\log(n))$ ancillae and $O(\log^*(n))$ GT gates. The second uses $O(n)$ ancillae and $O(1)$ GT gates.
Below we assume $n{\ge} 3$.
Let $OR_n{:} \, \{0,1\}^n \,{\to}\, \{0,1\}$ be the $n$-bit Boolean OR function.
Define an $n$-qubit unitary operator 
\[
\mathsf{OR}_n |x\ra = (-1)^{OR_n(x)} |x\ra. 
\]
It flips the phase of the $n$-qubit register $\ket{x}$ iff at least one bit of $x$ is non-zero.
Clearly, $\mathsf{OR}_n$ coincides with the multiply-controlled $\zgate$ gate up to relabeling of basis states $0$ and $1$
(which, in turn, coincides with the $n$-qubit Toffoli up to   a conjugation by the Hadamard gate on the target qubit).
Here we show how to implement $\mathsf{OR}_n$ by a quantum circuit composed of single-qubit gates and fractional $\cnotgate$ gates, that are then mapped into a small number of GT gates. 

Suppose $q\,{\ge}\,0$ is an integer.  Let $X_q$ be some fixed $2^q$-th root of Pauli $X$, such that $(X_q)^{2^q}\,{=}\,X$. One can choose the primary root,
\be
\label{eq:Xq}
X_q = H \left[ \ba{cc} 1 & 0 \\ 0& \exp{[i\pi/2^q]} \\ \ea \right]H,
\ee
where $H$ is the Hadamard gate. For example, $X_0\,{=}\,X$ is Pauli-$X$.  Let the fractional $\cnotgate$ gate be 
\[
CX_q = |0\ra\la 0|\otimes I + |1\ra\la 1|\otimes X_q.
\]
Our construction closely follows~\cite{hoyer2003quantum,takahashi2016collapse}. 
Choose an integer
$p\,{:=}\,\lceil \log_2{(n{+}1)}\rceil \,{\approx}\, \log{(n)}$.
Consider a bit string $x\in \{0,1\}^n$ with the Hamming weight $w(x)=x_1+x_2+\cdots +x_n$.  Note that $OR_n(x)\,{=}\,1$ iff 
\[
(X_q)^{w(x)}|0\ra = |1\ra \;\; \mbox{for some $q=0,1,\ldots,p{-}1$}.
\]
If $w(x)$ is odd then $(X_0)^{w(x)}|0\ra=\ket{1}$. Otherwise $w(x)$ is even and thus $w(x)/2$ is integer. If $w(x)/2$ is odd then $(X_1)^{w(x)}|0\ra=(X_0)^{w(x)/2}|0\ra=\ket{1}$.  Otherwise $w(x)\,{\equiv}\,0{\pmod 4}$ and thus $w(x)/4$ is integer.  If $w(x)/4$ is odd then $(X_2)^{w(x)}|0\ra=(X_0)^{w(x)/4}|0\ra=\ket{1}$, etc.  Thus we have the identity 
\begin{align}
\label{eq1}
(-1)^{OR_n(x)}=
\la 0^p | & (X_0^{w(x)}  \otimes  X_1^{w(x)} \otimes \cdots \otimes X_{p-1}^{w(x)} )^\dag
\,\mathsf{OR}_p\,
(X_0^{w(x)}   \otimes  X_1^{w(x)} \otimes \cdots \otimes X_{p-1}^{w(x)})|0^p\ra.
\end{align}

Consider a register with $n{+}p$ qubits where the last $p$ qubits are ancillae initialized in $|0\ra$.  One can implement the gates $X_q^{w(x)}$ controlled by the Hamming weight $w(x)$ of the first $n$ qubits and target on the last $p$ qubits using $n$ fractional $\cnotgate$ gates $CX_q$.  In fact, all $np$ fractional $\cnotgate$ gates can be implemented using one GT gate, since all controls are on the first $n$ qubits and all targets are on the last $p$ qubits, hence all commute 
(note that adding left and right layers of Hadamards on the ancillae transforms each $CX_q$ gate to $\czgate^a$ with $a\,{=}\,1/2^q$).
Next, we need to implement $\mathsf{OR}_p$ on the ancilla register of size $p$, which is nothing but an exponentially reduced version of the original problem of implementing $\mathsf{OR}_n$. One can thus recursively apply the reductions until only two qubits are left, after which $\mathsf{OR}_2$ can trivially be implemented using a single two-qubit gate. This results in the $2\log^*(n)-1$ GT construction of an $n$-qubit Toffoli using $\log(n)+O(\log(\log(n))$ ancillae.

To obtain an $O(1)$ GT circuit, we avoid the recursion. Instead, as soon as the problem is reduced to that of implementing the $\mathsf{OR}_p$, we invoke two-GT construction 
to obtain $\mathsf{OR}_p$, which uses $2^p{-}p{-}1$ ancillae. This is a GT-enabled modification of a well-known method~\cite{barenco1995elementary} to implement multi-controlled $\zgate$ gate that relies on computing all possible EXORs of $p$ Boolean inputs (which takes one GT gate) -- note the EXORs with a single literal do not need to be computed on the ancillae, since they are the input.  Once all possible EXORs are available, apply the single-qubit gate $\zgate^{1/2^{p-1}}$ to all EXORs to obtain $\mathsf{OR}_p$ (to obtain multi-controlled $\zgate$ odd and even weight EXORs experience the application of phases with opposite angles).
This costs no GT gates.  Finally, uncompute the EXORs using one GCZ gate.   The overall cost of implementing $n$-fold OR is thus $4$ (=1+(1{+}1)+1) GT gates using $2^{\lceil\log(n)\rceil}\,{-}\,1 < 2n$ ancillae.  In \appe{B} we show how to halve the gate counts in the above implementation using adaptive quantum circuits.  We note that the gate count $2\log^*(n)-1$, is in fact limited by the constant $9$, since $\log^*(n)$ cannot exceed the value of $5$ due to not enough elementary particles in the known universe to constitute qubits.  Thus, in effect both reported constructions carry a constant cost.

\section{Discussion}\label{sec:disc}

Our results significantly improve over state of the art, and for the most part settle asymptotic complexities.  We show that GT gates are a very powerful primitive for compiling short quantum circuits, since a serial approach would require $O(n^2/\log(n))$ \cite{patel2008optimal} and $O(n)$ gates respectively, for the Clifford and multiply-controlled Toffoli gates. Our results demonstrate the power of parallel SIMD (Single Instruction Multiple Data) operations in quantum computer settings, not unlike its classical, conventional computer counterparts. 

Our result also implies that GT gates enable more efficient generation of pseudo-random unitary operators known as unitary $t$-designs~\cite{low2010pseudo}.  In particular, $n$-qubit approximate unitary $t$-designs can be realized with $O(1)$ GT gates for any constant $t\,{=}\,O(1)$ since such designs can be realized using single-qubit gates
and at most $\tilde{O}(t^4)$ Clifford layers~\cite{haferkamp2020quantum}. 

This study sheds more light on space-time tradeoffs in quantum circuits. In particular, it turned out that the no-ancilla implementation of a Clifford unitary is possible without increasing the asymptotic complexity compared to a much simpler 
ancilla-enabled implementation.  An analogy can be drawn to the multiply-controlled Toffoli gate implemented using single-qubit and two-qubit gates, where it too turned out that a no-ancilla implementation is possible \cite{gidney2015} at the same asymptotic cost as a significantly more straightforward and practical ancilla-enabled implementation \cite{barenco1995elementary}.  Our study left the door open to discover a constant GT gate count implementation of the multiply-controlled Toffoli gate using few or no ancilla, as well as  implementation of
other commonly used unitary operations, such as the Quantum Fourier Transform.

\appendix

\section{Implementation of $n$-qubit Clifford operations}\label{appe:A}

In the main text we showed how to implement any $n$-qubit Clifford operation with $3|n$ using $25$ GCZ gates. Here we apply local optimizations to reduce the cost from $25$ to $20$ GCZ gates. 

First, let us show how to merge consecutive pairs of GCZ gates in the implementation
of the map $|v\ra \,{\mapsto}\, |Av\ra$ described in the main text.
First, consider the transformation
\begin{align*}
|x,y,z\ra & \xrightarrow{C_3(A_{00})} |A_{00}^{-1}z, y, A_{00}x \ra \\
& \xrightarrow{C_3(A_{22}A_{00})} |(A_{22}A_{00})^{-1}y, A_{22}z, A_{00}x\rangle \\ 
& \xrightarrow{W(A_{11}A_{22}A_{00})} 
|A_{00}x,A_{11}y, A_{22}z\ra.
\end{align*}
Denote the registers that hold the variables $x,y,z$ as $X,Y,Z$ respectively.
Then
the first $C_3$ gate is applied to $XZ$,
the second $C_3$ gate is applied to $XY$. 
We can merge the tailing GCZ from the first $C_3$ gate and the leading
GCZ from the second C3 gate. Indeed, these GCZs are used to implement $\cnotgate$ circuits
which have controls in $X$ and targets in $YZ$. 
The composition of these two $\cnotgate$ circuits is equivalent to a single GCZ gate
modulo conjugation by Hadamard gates on $YZ$.

As described in the main text, the $W(A)$ transformation is implemented as
\begin{align*}
|x,y,z\ra &\xrightarrow{C_3(B)} |B^{-1}y,Bx,z\ra \\
& \xrightarrow{C_3(D)} |B^{-1}y, D^{-1} z, DBx \ra \\
& \xrightarrow{C_3(B^{-1})} |B^{-1} DBx, D^{-1} z,y\ra \\
& \xrightarrow{C_3(D^{-1})} |z,D^{-1} B^{-1} DBx, y\ra \\ 
& = |z,Ax,y\ra.
\end{align*}
The first $C_3$ gate in the implementation of $W$ acts on $XY$. 
Since the second $C_3$ gate in the implementation of the map $|x,y,z\ra\to |A_{00}x,A_{11}y, A_{22}z\ra$
also acts on $XY$, we can merge
the tailing GCZ from this $C_3$ gate with the leading GCZ from the $W$ transformation. 
These two merges reduce the GCZ cost from $25$ to $23$. 

We can additionally merge consecutive pairs of GCZs in the implementation of $W$
by noting that 
\[
C_3(A)= \mathsf{CX}_{1,2}(A) \, \mathsf{CX}_{2,1}(A^{-1}) \,  \mathsf{CX}_{1,2}(A)
=\mathsf{CX}_{2,1}(A^{-1}) \, \mathsf{CX}_{1,2}(A) \,  \mathsf{CX}_{2,1}(A^{-1})
\]
for any invertible matrix $A$.
In other words, while implementing $C_3$, one has a freedom in choosing 
the control and the target register of the leading $\cnotgate$ gate (which is the same
for the tailing $\cnotgate$ gate).
We can exploit this freedom to realize $C_3$ gates that appear in the above implementation of $W$ as 
\[
\mbox{First $C_3$ gate:} \quad
 C_3(B) =   \mathsf{CX}_{X,Y}(A)\mathsf{CX}_{Y,X}(A^{-1})  \mathsf{CX}_{X,Y}(A),
\]
\[
\mbox{Second $C_3$ gate:} \quad
 C_3(D) =   \mathsf{CX}_{Z,Y}(D^{-1})\mathsf{CX}_{Y,Z}(D)  \mathsf{CX}_{Z,Y}(D^{-1}),
\]
\[
\mbox{Third $C_3$ gate:} \quad
 C_3(B^{-1}) =   \mathsf{CX}_{Z,X}(B^{-1})\mathsf{CX}_{X,Z}(B)  \mathsf{CX}_{Z,X}(B^{-1}), \text{ and}
\]
 \[
\mbox{Fourth $C_3$ gate:} \quad
 C_3(D^{-1}) =   \mathsf{CX}_{Y,X}(D)\mathsf{CX}_{X,Y}(D^{-1})  \mathsf{CX}_{Y,X}(D).
\]
Here the subscripts indicate registers acted upon by each $\cnotgate$ gate.
We can now merge GCZs that realize consecutive pairs 
\[
\{ \mathsf{CX}_{X,Y}(A),  \mathsf{CX}_{Z,Y}(D^{-1})\}, \text{ }
\{  \mathsf{CX}_{Z,Y}(D^{-1}),  \mathsf{CX}_{Z,X}(B^{-1})\},
\text{ and } \{ \mathsf{CX}_{Z,X}(B^{-1}),   \mathsf{CX}_{Y,X}(D)\}.
\]
These three merges reduce the GCZ cost further from $23$ to $20$.

\section{Implementation of the Toffoli gate by adaptive circuits}\label{appe:B}

Here we employ adaptive circuits to halve the GT gate count in the implementation of Toffoli-$n$ described in the main text.  Rather than uncomputing fractional $\cnotgate$s in Eq.~(9), one can instead measure each ancilla in the Pauli-$X$ basis immediately after applying $\mathsf{OR}_p$ (with $p\,{=}\,2$ or $\log(n)$, depending on the method of choice). Assuming $p\,{=}\,\log(n)$, the same arguments as above show that 
\[
(I_n \otimes \mathsf{OR}_p)  (X_0^{w(x)}  \otimes X_1^{w(x)}  \otimes  \cdots \otimes X_{p-1}^{w(x)})|x\ra \otimes |0^p\ra
= (\mathsf{OR}_n|x\ra)\otimes\left( X_0^{w(x)}  \otimes X_1^{w(x)} \otimes  \cdots \otimes X_{p-1}^{w(x)}|0^p\ra\right).
\]
Here $I_n$ is the identity operator on $n$ qubits. Suppose we  measure the $q$-th ancillary qubit in the $X$-basis and observe an outcome $m_q\,{\in}\, \{0,1\}$.  We need to check that this measurement does not introduce unwanted phase factors depending on $x$.  Using Eq.~(8) one can check that 
\begin{equation*}
\la m_q|H X_q^{w(x)}|0\ra=\frac1{\sqrt{2}} e^{i\pi m_qw(x)/2^q}
\end{equation*}
for any $m_q=0,1$. Thus all measurement outcomes are equally likely and the unwanted phase factor is a product of $e^{i\pi w(x)/2^q}$ over all ancillae $q$ such that $m_q{=}1$.  This phase factor can be cancelled by applying the single-qubit gate $\zgate^{1/2^q}$ 
to each of $n$ data qubits.  

By recursively applying the cancellation phase gates for each uncompute level after $\mathsf{OR}_2$ in the $O(\log^*(n))$-GT construction, one can then implement $\mathsf{OR}_n$ by  an adaptive circuit with $\log^*(n)$ GT gates and $\log(n)+O(\log(\log(n)))$ ancillae. For the $O(1)$-GT construction, note the $\mathsf{OR}_p$ circuit is of the form $U^{-1} L U$, acting on $p$ input and $n{-}p{-}1$ clean ancillary qubits, where the ancillae are returned in an all-zero state.  Here $U$ is a layer of $\cnotgate$ gates that computes EXORs of $p$ input variables into the ancilla and $L$ is a layer of single-qubit $Z$-rotations on the ancilla.  This implies the state $|\psi\ra$ of our quantum computer right after $L$ is
\begin{equation*}
|\psi\ra = (\mathsf{OR}_n|x\ra) \otimes U \left[\left( X_0^{w(x)}  \otimes X_1^{w(x)} \otimes  \cdots \otimes X_{p-1}^{w(x)}|0^p\ra\right)\otimes |0^{n-p-1}\ra\right].
\end{equation*}
Nominally, we are to apply to $|\psi\ra$ the operation $I_n \otimes [(H_p \otimes I_{n-p-1}) U^{-1}]$ prior to the measurements of $p$ qubits in the computational basis Z. 
Note that $U^{-1}$ returns each of the last $n{-}p{-}1$ qubits to the $|0\rangle$ state. Suppose we measure
each of these $|0\rangle$ qubits in the Pauli $X$-basis $\{|+\ra,|-\ra\}$, where $|+\ra=H|0\ra$ and 
$|-\ra=H|1\ra$. Clearly, such a measurement has no effect
on the remaining $p$ qubits. 
Consider some $\cnotgate$ that appears in $U^{-1}$.
Since the
target qubit of this $\cnotgate$ is measured in the Pauli $X$-basis, one can replace
the $\cnotgate$ by
the identity gate or
a single-qubit $\zgate$
gate acting on the control qubit by observing that
\[
\cnotgate |\phi\ra \otimes |+\ra = |\phi\ra \otimes |+\ra \quad \mbox{and} \quad 
\cnotgate |\phi\ra \otimes |-\ra = \zgate|\phi\ra \otimes |-\ra.
\]
Here $\phi$ is an arbitrary state of the control qubit.
Thus one can replace the entire circuit $U^{-1}$ by a product
of Pauli $\zgate$ gates applied to the first $p$
qubits. This saves one extra GT gate in the construction described above.
We conclude that one can implement $\mathsf{OR}_n$ by an adaptive circuit with two GT gates and $O(n)$ ancillae.

\end{document}